\RequirePackage{amsmath}
\documentclass[runningheads]{llncs}
\usepackage[T1]{fontenc}
\usepackage{amssymb}
\usepackage{geometry}
\usepackage{verbatim}
\usepackage{gastex} 
\usepackage{array,graphicx}
\usepackage{xcolor}
\usepackage[lowtilde]{url}
\DeclareSymbolFont{rsfscript}{OMS}{rsfs}{m}{n}
\DeclareSymbolFontAlphabet{\mathrsfs}{rsfscript}
\newcommand{\keywords}[1]{\par\addvspace\baselineskip
\noindent\keywordname\enspace\ignorespaces#1}

\newtheorem{criterion}{Criterion}
\spnewtheorem*{conjecture*}{Conjecture}{\bfseries}{\itshape}
\DeclareMathOperator{\rt}{rt}
\DeclareMathOperator{\DS}{DS}
\DeclareMathOperator{\Suff}{Suff}

\title{Algebraic synchronization criterion\\and computing reset words}
\titlerunning{Algebraic synchronization criterion and computing reset words}

\author{Mikhail~Berlinkov\inst{1} \and Marek~Szyku{\l}a\inst{2}} \institute{Institute of Mathematics
and Computer Science, Ural Federal University, Russia \and Institute
of Computer Science, University of Wroc{\l}aw, Poland}
\authorrunning{M.~Berlinkov and M.~Szyku{\l}a}

\begin{document}
\maketitle

\begin{abstract}
We refine a uniform algebraic approach for deriving upper bounds on reset
thresholds of synchronizing automata. We express the condition that an
automaton is synchronizing in terms of linear algebra, and obtain
upper bounds for the reset thresholds of automata with a short word
of a small rank. The results are applied to make several improvements in the area.

We improve the best general upper bound for reset thresholds of
finite prefix codes (Huffman codes): we show that an $n$-state
synchronizing decoder has a reset word of length at most $O(n \log^3
n)$. In addition to that, we prove that the expected reset threshold of a uniformly random synchronizing binary $n$-state decoder is at most $O(n \log n)$. We also show that for any non-unary alphabet there exist decoders whose reset threshold is in $\varTheta(n)$.

We prove the \v{C}ern\'{y} conjecture for $n$-state automata with a letter of rank at most $\sqrt[3]{6n-6}$. In another corollary, based on the recent results of Nicaud, we show that the probability that the \v{C}ern\'y conjecture does not hold for a random synchronizing binary automaton is exponentially small in terms of the number of states, and also that the expected value of the reset threshold of an $n$-state random synchronizing binary automaton is at most $n^{3/2+o(1)}$.

Moreover, reset words of lengths within all of our bounds are computable in polynomial time. We present suitable algorithms for this task for various classes of automata, such as (quasi-)one-cluster and (quasi-)Eulerian automata, for which our results can be applied.

\end{abstract}


\section{Introduction}

We deal with \emph{deterministic finite automata} (\emph{DFA})
$\mathrsfs{A} = (Q, \Sigma, \delta)$, where $Q$ is a non-empty
\emph{set of states}, $\Sigma$ is a non-empty \emph{alphabet}, and
$\delta\colon Q \times \Sigma \mapsto Q$ is the complete
\emph{transition function}. We extend $\delta$ to $Q \times
\Sigma^*$ and $2^Q \times \Sigma^*$ as usual, and for the image
(resp. preimage) of a set $S$ under a word $w$ we write shortly
$S.w$ (resp. $S.w^{-1}$). We denote $\Sigma^{\le c} = \{w \in
\Sigma^*\colon |w| \le c\}$, the set of all words over $\Sigma$ of
length at most $c$. The empty word is denoted by $\varepsilon$.
Throughout the paper, by $n$ we denote the cardinality $|Q|$, and by
$k$ we denote $|\Sigma|$.

A word $w$ \emph{compresses} a subset $S \subseteq Q$ if $|S.w| <
|S|$. Then we say that $S$ is \emph{compressible}. The \emph{rank}
of a word $w$ is $|Q.w|$. A \emph{reset word} or a
\emph{synchronizing word} is a word $w \in \Sigma^*$ of rank $1$,
that is, $w$ takes the automaton to a particular state no matter of
the current state. An automaton is called \emph{synchronizing} if it
possesses a reset word. An example of a synchronizing automaton from the \v{C}ern\'{y} series~\cite{Cerny1964} is presented in Fig.~\ref{fig:4rcp}(left). One can verify that its shortest reset word is $ba^3ba^3b$. The length of the shortest reset word is called the \emph{reset threshold} and is denoted by $\rt(\mathrsfs{A})$.

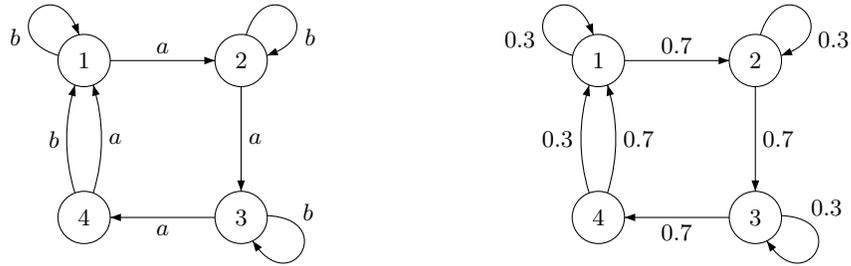
\begin{figure}[ht]
 \begin{center}
  \unitlength=3.3pt
 \begin{picture}(18,26)(20,-3)
    \gasset{Nw=6,Nh=6,Nmr=3,loopdiam=5}
    \thinlines
    \node(A1)(0,18){$1$}
    \node(A2)(18,18){$2$}
    \node(A3)(18,0){$3$}
    \node(A4)(0,0){$4$}
    \drawloop[loopangle=135,ELpos=30](A1){$b$}
    \drawloop[loopangle=45,ELpos=70](A2){$b$}
    \drawloop[loopangle=-30,ELpos=25](A3){$b$}
    \drawedge(A1,A2){$a$}
    \drawedge(A2,A3){$a$}
    \drawedge(A3,A4){$a$}
    \drawedge[curvedepth=2](A4,A1){$b$}
    \drawedge[ELside=r, curvedepth=-2](A4,A1){$a$}
 \end{picture}
 \begin{picture}(18,26)(-20,-3)
   \gasset{Nw=6,Nh=6,Nmr=3,loopdiam=5}
    \thinlines
    \node(A1)(0,18){$1$}
    \node(A2)(18,18){$2$}
    \node(A3)(18,0){$3$}
    \node(A4)(0,0){$4$}
    \drawloop[loopangle=135,ELpos=30](A1){$0.3$}
    \drawloop[loopangle=45,ELpos=70](A2){$0.3$}
    \drawloop[loopangle=-30,ELpos=25](A3){$0.3$}
    \drawedge(A1,A2){$0.7$}
    \drawedge(A2,A3){$0.7$}
    \drawedge(A3,A4){$0.7$}
    \drawedge[curvedepth=2](A4,A1){$0.3$}
    \drawedge[ELside=r, curvedepth=-2](A4,A1){$0.7$}
 \end{picture}
 \end{center}
 \caption{The automaton $\mathrsfs{C}_4$ and the associated Markov chain for $P(a) = 0.7, P(b) = 0.3$.}
 \label{fig:4rcp}
\end{figure}

Synchronizing automata serve as transparent and natural models of
various systems in many applications in various fields, such as coding theory, DNA-computing,
robotics, testing of reactive systems, and theory of information
sources. They also reveal interesting connections with symbolic
dynamics, language theory, and many other parts of mathematics. For a detailed introduction to the theory of synchronizing automata we
refer the reader to the
surveys~\cite{KariVolkov2013Handbook,Volkov2008Survey}, and for a
review of relations with coding theory to~\cite{Jurgensen2008}.

In various applications, reset words allow to reestablish the control under the system modeled by an automaton. The reset threshold of an automaton serves as a natural measure of synchronization.
Naturally, the shorter reset word the better, so from both theoretical and practical points of view it is important to compute the reset threshold, and shortest or short enough reset words.

In~1964 \v{C}ern\'{y}~\cite{Cerny1964} constructed for each $n>1$ a
synchronizing automaton $\mathrsfs{C}_n$ with $n$ states and 2 input
letters whose reset threshold is $(n-1)^2$. The automaton
$\mathrsfs{C}_4$ is shown in Fig.~\ref{fig:4rcp}(left). Soon after
that he conjectured that those automata represent the worst possible
case, thus formulating the following hypothesis:
\begin{conjecture*}[\v{C}ern\'y]
\label{Cerny_Conj} Each synchronizing automaton $\mathrsfs{A}$ with
$n$ states has a reset word of length at most $(n-1)^2$, i.e.\
$\rt(\mathrsfs{A}) \le (n-1)^2$.
\end{conjecture*}
By now, this simply looking conjecture is arguably the most
longstanding open problem in the combinatorial theory of finite
automata. Moreover, the best upper bound known so far for the reset threshold of a synchronizing $n$-state automaton is equal to
$\frac{n^3-n}6-1$ (for $n\ge 4$) and so is cubic in $n$ (see
Pin~\cite{Pin1983}). Thus it is of certain importance to prove better specific upper bounds for various important classes of synchronizing automata.

In this paper, we improve several results concerning reset thresholds.
First, we express the condition that an automaton is synchronizing in terms of linear algebra, and derive upper bounds
for automata with a word of a small rank (Section~\ref{sec:criterion}).
Then, we apply the results to improve upper bounds in several cases.

We apply the results to improve upper bounds in several cases.
In Section~\ref{sec:cerny_random} we
show that the \v{C}ern\'{y} conjecture holds for automata with a
letter of rank $\sqrt[3]{6n-6}$, which improves the previous
logarithmic bound \cite{Pin1972Utilisation}. Also, basing on the
recent results of Nicaud
\cite{Nicaud2014FastSynchronizationOfRandomAutomata}, we show that the \v{C}ern\'{y} conjecture holds for a random synchronizing binary automaton with probability
exponentially (in $n$) close to 1, and that the expected reset threshold is
at most $n^{3/2 + o(1)}$.

The next important application of our results is an upper bound for
the length of the shortest reset words of decoders of finite prefix codes
(Huffman codes), which are one of the most popular methods of data compression.
One of the problems with compressed data is the reliability in case of presence of errors in the compressed text. Eventually, a single error may possibly destroy the whole encoded string. One of the solutions to this problem (for Huffman codes) is a use of codes whose decoders can be synchronized by a reset word, regardless of the possible errors. Then, by inserting reset words into compressed data, we make the data error-resistant to some extent.

The reset thresholds of binary Huffman codes was first studied by Biskup and Plandowski \cite{Biskup2008HuffmanCodes,BiskupPlandowski2009HuffmanCodes}, who proved a general upper bound of order $O(n^2 \log n)$.
They also showed that a word of this length can be computed in polynomial time.
The bound was later improved to $O(n^2)$ for a wider
class of \emph{one-cluster} automata \cite{BBP2011QuadraticUpperBoundInOneCluster}.
In Section~\ref{sec:prefix_codes} we prove an upper bound of order $O(n \log^3 n)$.
Next, we consider random decoders, and show that the expected reset threshold of a uniformly random synchronizing $n$-state binary decoder is at most $O(n \log n)$.
We also show a series of decoders with linear reset thresholds over any alphabet of size at least 3. Such series were known only for a binary alphabet \cite{BiskupPlandowski2009HuffmanCodes}.

Unlike the general case, the \v{C}ern\'{y} conjecture has been
approved for various classes of automata such as
circular~\cite{Dubuc1998}, Eulerian~\cite{Kari2003Eulerian}, and
one-cluster automata with prime length
cycle~\cite{Steinberg2011OneClusterPrime}. Later, specific quadratic
upper bounds for some generalizations of these classes were obtained
in~\cite{BBP2011QuadraticUpperBoundInOneCluster,BP2009QuadraticUpperBoundInOneCluster,Berlinkov2013QuasiEulerianOneCluster}.
However, no efficient algorithm for finding reset words with lengths
within the specified bounds has been presented for these classes.
Moreover, there is no hope to get a polynomial algorithm for finding
the shortest reset words in the general case, since this problem has
been shown to be $\mathrm{FP}^{\mathrm{NP}[\log]}$-hard
\cite{OM2010}. Also, unless $\mathrm{P}=\mathrm{NP}$, there is no
polynomial algorithm for computing the reset threshold for a given
automaton within the approximation ratio $n^{1-\varepsilon}$ for any $\varepsilon>0$, even in the case of a binary
alphabet~\cite{GawrychowskiStraszak2015StrongInapproximabilityOfTheShortestResetWord}
(cf.~also
\cite{Berlinkov2014Approximating,Berlinkov2014OnTwoAlgorithmicProblems,GH2011}).

In Section~\ref{sec:algorithms} we present polynomial algorithms for
finding reset words of length guaranteed to be within the proven bounds. In particular, our algorithms can be applied to the classes of decoders of finite prefix codes, and also to generalized classes of quasi-Eulerian and quasi-one-cluster automata. Since from our results it is possible to derive the bounds
from~\cite{BBP2011QuadraticUpperBoundInOneCluster,BP2009QuadraticUpperBoundInOneCluster,Berlinkov2013QuasiEulerianOneCluster,CarpiDAlessandro2013IndependendSetsOfWords,Kari2003Eulerian,Steinberg2011AveragingTrick,Steinberg2011OneClusterPrime}), our algorithms apply to them as well.

A preliminary version of some of these results previously appeared in~\cite{BS2015AlgebraicSynchronizationCriterion}.


\section{Algebraic synchronization criterion}
\label{sec:criterion}

In this section we refine some results
from~\cite{Berlinkov2013QuasiEulerianOneCluster}, formulate the algebraic synchronization criterion, and derive upper bounds for reset thresholds of automata with a word of a small rank.
For this purpose, we associate a natural linear structure with an automaton
$\mathrsfs{A}$. By $\mathbb{R}^n$ we denote the real $n$-dimensional
linear space of row vectors. Without loss of generality, we assume
that $Q=\{1,2,\dots,n\}$ and then assign to each subset $K\subseteq
Q$ its \emph{characteristic vector} $[K] \in \mathbb{R}^n$, whose
$i$-th entry is 1 if $i \in K$, and 0, otherwise. For $q\in Q$ we
write $[q]$ instead of $[\{q\}]$ to simplify the notation. By
$\langle S\rangle$ we denote the linear span of
$S\subseteq\mathbb{R}^n$. The $n \times n$ identity matrix is
denoted by $I_n$.

Each word $w \in \Sigma^*$ corresponds to a linear transformation of
$\mathbb{R}^n$. By $[w]$ we denote the matrix of this transformation
in the standard basis $[1],\ldots,[n]$ of $\mathbb{R}^n$. For
instance, if $\mathrsfs{A}=\mathrsfs{C}_4$ from
Figure~\ref{fig:4rcp} (left), then
$$[a]=\left(
  \begin{smallmatrix}
    0 & 1 & 0 & 0\\
    0 & 0 & 1 & 0\\
    0 & 0 & 0 & 1\\
    1 & 0 & 0 & 0
  \end{smallmatrix}
\right),\ [b]=\left(
  \begin{smallmatrix}
    1 & 0 & 0 & 0\\
    0 & 1 & 0 & 0\\
    0 & 0 & 1 & 0\\
    1 & 0 & 0 & 0
  \end{smallmatrix}
\right),\ [ba]=\left(
  \begin{smallmatrix}
    0 & 1 & 0 & 0\\
    0 & 0 & 1 & 0\\
    0 & 0 & 0 & 1\\
    0 & 1 & 0 & 0
  \end{smallmatrix}
\right).$$
Clearly, the matrix $[w]$ has exactly one non-zero entry
in each row. In particular, $[w]$ is \emph{row stochastic}, that is,
the sum of entries in each row is equal to $1$. In virtue of
row-vector notation (apart
from~\cite{Berlinkov2013QuasiEulerianOneCluster}), we get that
$[uv]=[u][v]$ for every two words $u,v \in \Sigma^{*}$. By $[w]^T$
we denote the transpose of the matrix $[w]$. One easily verifies
that $[S.w^{-1}]=[S][w]^T$. Let us also notice that within this
definition the (adjacency) matrix of the underlying digraph of
$\mathrsfs{A}$ is equal to $\sum_{a \in \Sigma}[a]$.

Recall that a word $w$ is a reset word if $q.w^{-1}=Q$, for some
state $q \in Q$. Thus, in the language of linear algebra, we can
rewrite this fact as $[q][w]^T = [Q]$. For two vectors $g_1,g_2 \in
\mathbb{R}^n$, we denote their usual inner (scalar) product by
$(g_1,g_2)$. We say that a vector (matrix) is \emph{positive}
(\emph{non-negative}) if it contains only positive (non-negative)
entries. Let $p \in \mathbb{R}^n_+$ be a positive row stochastic
vector. Then $([Q],p)=1$, and a word $w$ is a reset word if and only
if $$([q.w^{-1}],p) = ([q][w]^T, p) = ([q], p [w]) = 1.$$

Now we need to recall a few properties of Markov chains. A
\emph{Markov chain} of an automaton $\mathrsfs{A}$ is the random
walk process of an agent on the underlying digraph of $\mathrsfs{A}$
where each time an edge labeled by $a_i$ is chosen according to a
given probability distribution $P\colon \Sigma \mapsto R$. The
matrix $S(\mathrsfs{A},P)=\sum_{i=1}^{k}{P(a_i)[a_i]}$ is called the
\emph{transition matrix} of this Markov chain. An example of a
Markov chain associated with the automaton
$\mathrsfs{A}=\mathrsfs{C}_4$ is presented in Figure~\ref{fig:4rcp}
(right) for $P(a) = 0.7, P(b) = 0.3$. A non-negative square matrix
$M$ is \emph{primitive} if for some $d>0$, the matrix $M^d$ is
positive. Call a finite set of words $W$ \emph{primitive} if the
sum of the matrices of words from $W$ is primitive. It is well
known that if $\mathrsfs{A}$ is strongly connected and
synchronizing, then the matrix of the underlying digraph of
$\mathrsfs{A}$ is primitive, and so is the matrix of a Markov chain
of $\mathrsfs{A}$ for any positive probability distribution $P$ (see
e.g.~\cite{AGV2010,AGV2013,Berlinkov2013QuasiEulerianOneCluster}).

\begin{proposition}\label{prop:markov}
Let $M$ be a row stochastic $n \times n$ matrix. Then there exists a
\emph{stationary distribution} $\alpha \in\mathbb{R}^{n}$, that is,
a non-negative stochastic vector satisfying $\alpha M = \alpha$.
Moreover, if $M$ is primitive then $\alpha$ is unique and positive.
\end{proposition}

Call a set of words $W \subseteq \Sigma^*$ \emph{complete} for a
subspace $V \leq \mathbb{R}^n$, with respect to a vector $g \in V$,
if
$$\langle g [w] \mid w \in W \rangle=V.$$
For a subset $S \subseteq Q$ we define $V_S = \langle[p] \mid
p \in S\rangle \leq \mathbb{R}^n$.

We aim to
strengthen~\cite[Theorem~9]{Berlinkov2013QuasiEulerianOneCluster}.
Namely, we show that the condition that $\mathrsfs{A}$ is
synchronizing is not necessary if we require completeness for the
corresponding set of words, and that only completeness with respect
to the stationary distribution of $\mathrsfs{A}$ is required. As
in~\cite{Berlinkov2013QuasiEulerianOneCluster} we construct an
auxiliary automaton. We fix two positive integers $d_1, d_2$ and two
non-empty sets of words $W_1 \subseteq \Sigma^{\le d_1}$, $W_2 \subseteq \Sigma^{\le d_2}$. Consider the automaton
$$\mathrsfs{A}_c(W_1,W_2)=(R, W_2 W_1, \delta_{\mathrsfs{A}_c}),$$
where $R = \{q.w \mid q \in Q,\ w \in W_1 \}$ and $W_2 W_1 = \{w_2
w_1 \in \Sigma^* \mid w_2 \in W_2, w_1 \in W_1\}$. The transition
function $\delta_{\mathrsfs{A}_c}$ is defined in compliance with the
actions of words in $\mathrsfs{A}$, i.e.
$\delta_{\mathrsfs{A}_c}(q,w) = \delta(q,w)$, for all $q \in R$ and
$w \in W_2 W_1$. Note that $\delta_{\mathrsfs{A}_c}$ is well defined
because $q.w \in R$ for all $q \in Q$ and $w \in
\Sigma_{\mathrsfs{A}_c}$. Without loss of generality we may assume
that $R=\{1,2,\ldots,|R|\}$ where $r = |R|$.

Let $P_1$ and $P_2$ be some positive probability distributions on
the sets $W_1$ and $W_2$, respectively, and denote $[P_i] = \sum_{w
\in W_i}{P_i(w)[w]}$ for $i=1,2$. Then the $r \times r$ submatrix
formed by the first $r$ rows and the first $r$ columns of the matrix
$$S(\mathrsfs{B},P_2P_1)=[P_2][P_1] = \sum_{w_1\in W_1, w_2 \in W_2}{P_1(w_1)P_2(w_2)[w_2][w_1]}$$
is the transition matrix of the Markov chain on ${\mathrsfs{A}_c}$.
By Proposition~\ref{prop:markov} there exists a steady state
distribution $\alpha = \alpha({\mathrsfs{A}_c}) \in V_R$, that is, a
stochastic vector (with first $r$ non-negative entries) satisfying
$\alpha S({\mathrsfs{A}_c},P_2P_1) = \alpha$.

For a vector $g \in \mathbb{R}_+^n$, by $\DS(g)$ we denote the number
of different positive sums of entries of $g$, i.e.
$$\DS(g) = |\{ (g,z) \mid z \in \{0,1\}^n \}| - 1.$$

\begin{theorem}\label{thm:main_ext}
Let $\mathrsfs{A} = (Q,\Sigma,\delta)$ be an automaton and let
$$\mathrsfs{B}=\mathrsfs{A}_c(W_1,W_2) = (R, W_2 W_1, \delta_\mathrsfs{B}),$$
be the automaton defined as above. If $W_2 W_1$ is complete for $V_R$
with respect to $\alpha$, and $w_0 \in \Sigma^*$ is a word with
$Q.w_0 = R$, then:
\begin{enumerate}
  \item \label{i_2} If $x \in V_R \setminus \langle[R]\rangle$, then there exists $w \in W_2 W_1$ such that $(x,\alpha[w]) > (x,\alpha)$;
  \item \label{i_1} $\mathrsfs{B}$ is synchronizing and $\rt(\mathrsfs{B}) \leq \DS(\alpha)-1$;
  \item \label{i_3} $\mathrsfs{A}$ is synchronizing and
  $$\rt(\mathrsfs{A}) \leq \begin{cases}
|w_0| + \rt(\mathrsfs{B})(d_1 + d_2) \leq |w_0| + (\DS(\alpha) - 1)(d_1 + d_2) & \text{if $R \neq Q$,}\\
1 + (\DS(\alpha) - 2)(d_1 + d_2) & \text{if $R = Q$.}\\
\end{cases}$$
\end{enumerate}
\end{theorem}
\begin{proof}
Let $x \in V_R \setminus \langle[R]\rangle$. We have
\begin{equation}
\label{eq_xq} (x,[q]) \neq (x,\alpha) \text{ for some } q \in R.
\end{equation}

Since $[q] \in V_R$ and $W_2 W_1$ is complete for $V_R$ with respect
to $\alpha$, we can represent it as follows:
\begin{equation}\label{eq:represent_g}
[q]=\sum_{w_1 \in W_1, w_2 \in W_2}{\lambda_{w_1,w_2}
\alpha[w_2][w_1]}\text{ for some } \lambda_{w_1,w_2} \in \mathbb{R}.
\end{equation}

Multiplying (\ref{eq:represent_g}) by the vector $[Q]$ we obtain
\begin{equation}\label{eq_mult_by_Q}
1=([q],[Q])=\left(\sum_{w_1 \in W_1, w_2 \in W_2}{\lambda_{w_1,w_2}
\alpha[w_2][w_1]},[Q]\right)= \sum_{w_1 \in W_1, w_2 \in
W_2}{\lambda_{w_1,w_2}}.
\end{equation}

Multiplying (\ref{eq:represent_g}) by the vector $x$ we obtain
\begin{equation}
\label{eq_mult_by_X}([q],x)=\left(\sum_{w_1 \in W_1, w_2 \in
W_2}{\lambda_{w_1,w_2} \alpha[w_2][w_1]},x\right).
\end{equation}

Arguing by contradiction, suppose $(x, \alpha[u_2][u_1]) =
(x,\alpha)$ for every $u_1 \in W_1$, $u_2 \in W_2$. Then by
(\ref{eq_mult_by_Q}) and~(\ref{eq_mult_by_X}) we get that
$([q],x)=(x,\alpha)$ contradicts (\ref{eq_xq}). Hence
$$(x, \alpha[u_2][u_1]) \neq (x,\alpha),$$
for some $u_1 \in W_1$, $u_2 \in W_2$.

Since $\alpha[P_2][P_1] = \alpha$, we have either
$(x,\alpha[u_2][u_1])>(x, \alpha)$ or
$(x,\alpha[v_2][v_1])>(x,\alpha)$ for some other $v_1 \in W_1, v_2
\in W_2$. Thus Claim~\ref{i_2} follows.

The proof of Claims~\ref{i_1} and~\ref{i_3} follows from an application of the \emph{greedy extension algorithm} from Section~\ref{sec:algorithms}.
\qed
\end{proof}

The following properties are easily verified and will be useful:

\begin{remark}
If $W_2$ is complete for $\mathbb{R}^n$ with respect to some vector $g$, then $W_2 W_1$ is complete for $V_R$ with respect to $g$.
\end{remark}

\begin{remark}
If for some positive probability distributions on $W_2$ and $W_1$, the set $W_2 W_1$ is complete for $V_R$ with respect to each stationary distribution, then $\mathrsfs{B}=\mathrsfs{A}_c(W_1,W_2)$ is strongly connected and synchronizing.
\end{remark}

\begin{remark}
If $\mathrsfs{B}=\mathrsfs{A}_c(W_1,W_2)$ is strongly connected and $W_2 W_1$ is complete for $V_R$ with respect to a stationary distribution induced by some positive probability distributions on $W_2$ and $W_1$, then $W_2 W_1$ is complete for $V_R$ with respect to any stochastic vector.
\end{remark}

\begin{criterion}\label{crit:syn}
Let $\alpha$ be a stationary distribution of the Markov chain
associated with a strongly connected $n$-state automaton
$\mathrsfs{A}$ by a given positive probability distribution $P$ on
the alphabet $\Sigma$. Then $\mathrsfs{A}$ is synchronizing if and
only if there exists a set of words $W$ which is complete for
$\mathbb{R}^n$ with respect to $\alpha$.
\end{criterion}
\begin{proof}
If $\mathrsfs{A}$ is synchronizing then for each state $q \in Q$
there is a reset word $w_q$ such that $Q.w_q = q$. Hence, $W = \{w_q
\mid q\in Q\}$ is complete for $\mathbb{R}^n$ with respect to
$\alpha$, because $\alpha [w_q] = [q]$.

Let us prove the opposite direction. Set
$$W_1 = \{\varepsilon\},\ W_2 = \Sigma^{\leq n-1}, \text{ and } [P_2] =
\frac{1}{n}\sum_{i=0}^{n-1}[P]^{i}.$$
Then $\alpha [P_2] = \alpha$,
and $W_2$ is complete for $\mathbb{R}^n$ with respect to $\alpha$.
Hence $\mathrsfs{A}$ is synchronizing by Theorem~\ref{thm:main_ext}.
\qed
\end{proof}

It is worth mentioning that an equivalent criterion with respect to our case has been independently obtained in~\cite{VP2015} in terms of affine operators via a so called \emph{fixed point} approach.

Now we can provide an upper bound for the reset threshold, if we can find a short word of a small rank.

\begin{theorem}\label{thm:about_threshold}
Let $\mathrsfs{A} = (Q,\Sigma,\delta)$ be a synchronizing automaton.
Then there is a unique (strongly connected) \emph{sink component} $\mathrsfs{S} = (S,\Sigma,\delta)$.
Let $w$ be a word and denote $r = |Q.w|$. Let $0 < d < n$ be the smallest positive integer such that $\Sigma^{\leq d}$ is complete for $V_S$ with respect to any stochastic vector $g \in V_S$ and
for each $q \in Q$ there is a word $u_q \in \Sigma^{\leq d}$ such
that $q.u_q \in S \cap Q.w$. Then
$$\rt(\mathrsfs{A}) \leq \begin{cases}
(|w| + d)(\frac{r^3 - r}{6})-d & \text{if $r \ge 4$;}\\
|w|+(|w|+d)(r-1)^2 & \text{if $r \le 3$.}\\
\end{cases}$$
Moreover, any pair of states from $Q$ is compressible by a word of
length at most $|w|+(|w|+d)\frac{r^2-r}{2}$.
\end{theorem}
\begin{proof}
Let $W_1 = \{w\}$, $W_2 = \Sigma^{\leq d}$, $w_0 = w$, and let
$P_1$, $P_2$ be arbitrary positive distributions on $W_1$ and $W_2$,
respectively. We define $\mathrsfs{B} = \mathrsfs{A}_c(W_1,W_2)$ as
in Theorem~\ref{thm:main_ext}, and consider its sink component
$\mathrsfs{C} = \mathrsfs{S}_c(W_1,W_2) = (Q_C, \Sigma, W_2 W_1)$.
Clearly $Q_C = Q.w \cap S$, and $W_2 W_1$ is complete for $V_{Q_C}
\leq V_S$ with respect to any stochastic vector $g \in V_{Q_C}$. By
Criterion~\ref{crit:syn} we obtain that $\mathrsfs{C}$ is
synchronizing.

Since for each $q \in Q.w$ there is a word $u_q \in W_2$ and so $w_q
\in W_2 W_1$ (a letter of $\mathrsfs{B}$) which takes $q$ to $Q_C$,
the automaton $\mathrsfs{B}$ is synchronizing.

Since $\mathrsfs{B}$ is synchronizing, $|Q.w_0| = r$, and $|u| \leq
|w| + d$ for each $u \in W_2 W_1$, we have that $\rt(\mathrsfs{A}) \leq |w| + \rt(\mathrsfs{B})(|w|+d)$.
By Pin's bound for the reset threshold in the general case
\cite{Pin1983}, $\rt(\mathrsfs{B}) \le \frac{r^3-r}{6}-1$ for $r \ge
4$.

Since $\mathrsfs{B}$ is synchronizing and there are
$\frac{r^2-r}{2}$ pairs in $Q.w$, any pair of states in $Q$ can be
compressed by a word of length at most $|w|+(|w|+d)\frac{r^2-r}{2}$.
\qed
\end{proof}


\section{Finding reset words of lengths within the bounds}
\label{sec:algorithms}

Throughout this section suppose we are given a strongly connected
automaton $\mathrsfs{A}$, a word $w_0$ such that $Q.w_0 = R$ for
some $R \subseteq Q$, a non-empty polynomial set of words $W_1$ with a positive distribution $P_1$, and
a set of words $W_2$ with a positive distribution $P_2$, which satisfy Theorem~\ref{thm:main_ext}.

Consider the case when $W_2$ is of polynomial size.
Then we can calculate the dominant eigenvector $\alpha \in \mathbb{R}^n$ of the matrix $[P_2][P_1]$. Under certain assumptions on rationality of the
distributions, it can be done in polynomial time.
Next, depending on whether the bound is obtained by Theorem~\ref{thm:about_threshold} or Claim~\ref{i_1} of Theorem~\ref{thm:main_ext}, we use either a greedy compressing
algorithm (such as in~\cite{Ep1990}), or the following \emph{greedy extension algorithm}, respectively.

\textbf{The Greedy Extension Algorithm}. We start from $x_0 = [q]$
for $q \in R$ and by~Claim~\ref{i_1} of~Theorem~\ref{thm:main_ext}
find $u_0 \in W_2W_1$ such that $(x_0,\alpha[u_0]) > (x_0,\alpha)$.
For $i=0, 1, \ldots $ following this way until $x_i \in
\langle[R]\rangle$, find for $x_{i+1} = x_i [u_i]^t$ a word $u_{i+1}
\in W_2W_1$ such that $(x_{i+1},\alpha [u_{i+1}]) >
(x_{i+1},\alpha)$. Since $x_i$ is a 1-0 vector, we need at most
$\DS(\alpha)-1$ steps until $x_i = [q] ([u_i u_{i-1} \dots u_0])^t = [R]$. As the result we return the word $w_0 u_i u_{i-1} \dots u_0$.
Notice that in the case when $R=Q$ we can choose $q$ such that for some letter $a \in \Sigma$, we have $|q.a^{-1}| > 1$ and set $u_0 = a$.
\qed
\smallskip


The problem is that usually $W_2$ is given by $\Sigma^{\leq d}$ for
some $d = \mathrm{poly}(n)$.
The following reduction procedure allows to replace potentially exponential set $W_2$ with a polynomial set of words $W$, whose the longest words are not longer than those of $W_2$.



\textbf{The Reduction Procedure}.
The procedure takes a number $d$, and returns a polynomial subset $W \subseteq \Sigma^{\leq d}$
such that $\langle W \rangle = \langle \Sigma^{\leq d} \rangle$ and the maximum length of words from $W$ is the shortest possible.

We start with $V_0 = \{ I_n \}$ and $W = \{\varepsilon\}$.
In each iteration $i \in \{1,2,\ldots\}$ we first set $V_{i+1} = V_i$.
Then we subsequently check each letter $a \in \Sigma$ and each word $u \in W$ of length $i$:
If the matrix $[u a]$ does not belong to the subspace $V_{i+1}$, we add the word $u a$ to $W$ and the matrix $[u a]$ to the basis of $V_{i+1}$.
We stop the procedure at the first iteration where nothing is added.

Since in an $i$-th iteration we have considered $a \in \Sigma$ and $u \in W$ of length less than $i$ in the previous iterations, by induction we get
$$V_{i} = \langle I_n (W \cap \Sigma^{\leq i}) \rangle = \langle I_n \Sigma^{\leq i} \rangle.$$
It follows from the ascending chain argument (see e.g.~\cite{Steinberg2011OneClusterPrime,Kari2003Eulerian}) that for
some $j < n$ we have
$$ V_{j} = V_{j+1} = \dots .$$
Thus the procedure is stopped at the first such $j$, and $j \leq \min\{d,n-1\}$.
We get that $\langle W \rangle = V_j = \langle \Sigma^d \rangle$.
Since in each step we add only independent matrices as the basis of $V_{i+1}$, we get $|W| = \dim(V_j)$.
Also the lengths of words in $W$ are at most $j \leq \min\{d,n-1\}$.
\qed
\smallskip

Using the reduction procedure for total completeness we can replace
$\Sigma^d$ from Theorem~\ref{thm:about_threshold} by a polynomial
$W$, which is also complete for $V_S$ with respect to any
stochastic vector $g \in V_S$. Hence, this yields a polynomial time
algorithm finding reset words of lengths within the bound of
Theorem~\ref{thm:about_threshold}.

In some situations we are interested only in completeness with
respect to a given vector $\alpha$. Then we can find a reduced set
$W$ of potentially shorter words than that obtained by the general reduction
procedure.


\textbf{The Reduction Procedure for $\alpha$-Completeness}.
The procedure takes a number $d$ and a vector $\alpha$,
and returns a polynomial subset $W \subseteq \Sigma^{\leq d}$ such that $\langle
\alpha W \rangle = \langle \alpha \Sigma^{\leq d} \rangle$ and the maximum length of words from $W$ is the shortest possible.

We just follow the general reduction procedure, where instead of matrix spaces we consider vector spaces.
It is enough to replace $I_0$ by $\alpha$, and we obtain $\langle \alpha W \rangle = V_j = \langle \alpha \Sigma^{\leq d} \rangle$.
\qed
\smallskip

\begin{remark}
Instead of $\Sigma^{\leq d}$ the reduction procedures can also reduce any set of words $W' \subset \Sigma^*$ that is factor-closed.
A set of words $W'$ is \emph{factor-closed} if $uvw \in W'$ implies that $uw \in W'$, for each $u,v,w \in \Sigma^*$.

This follows since the ascending chain argument still holds.
If $V_i = V_{i+1}$ then also $V_i = \langle W' \rangle$.
Assume for a contrary that $V_i < \langle W' \rangle$, and let $ua \in W'$ be a shortest word such that $I_n [ua]$ is independent to $I_n W$.
Then there is some prefix $w \in W$ of $u$, and $u = wva$.
Since $u$ was a shortest word, $wv$ is dependent to $I_n W$, so $$\langle I_n W \rangle = \langle I_n (W \cup [wv]) \rangle < \langle I_n (W \cup [wva]) = \langle I_n (W \cup [wa]).$$
Since $W'$ is factor-closed, $wa \in W'$, and it was considered in the reduction procedure and added to $W$ -- a contradiction.
\qed
\end{remark}


The following procedure finds a polynomial subset $W \subseteq W_2$ such that $W W_1$ is still primitive under the restriction to $V_R$, and the words in $W$ are as short as possible.


\textbf{The Reduction Procedure for Primitive Sets}.
As the input the procedure takes a set of words $W_1$ and a number $d > 0$ such that $\Sigma^{\leq d}
W_1$ is primitive when restricted to $V_R$ where $R = Q.W_1$, and
returns a polynomial subset $W \subseteq \Sigma^{\leq d}$ such that $W W_1$ is
also primitive for $R$.

We follow the reduction procedure with the following modification:
Instead of adding a word $u a$ to $W$ if $[u a]$ does not belong to the current subspace $V_{i+1}$, we add $u a$ if for some $w_1 \in W_1$ there is a non-zero entry in $[u a][w_1]$ such that this entry is zero in all matrices $[w]$ for $w \in W W_1$.
We stop the procedure as soon as the set of words $W W_1$ restricted to $V_R$ becomes primitive.

To check whether $W W_1$ is primitive, since the exponent of $r \times r$ primitive matrix is at most $(r-1)^2+1$ (see~\cite{AGV2013}), it is enough to check that the $((r-1)^2+1)$-th power of the sum of all matrices $[w]$ for $w \in W W_1$ is positive. Since in each step we make positive at least one entry in this sum, we need at most $(r-1)^2+1$ steps in total.
\qed
\smallskip




Now, given some sets $W_1$ and $W_2 = \Sigma^{\leq d}$, we can first find $W \subseteq W_2$ such that $W W_1$ is primitive for $V_R$. Then, we can choose some positive probability distribution on $W$, which induces a unique stationary distribution $\beta$.
We can also find $W' \subseteq W_2$ complete with respect to $\beta$.
The problem here is that, for the set $(W \cup W') W_1$ there is possibly no positive probability distribution inducing the stationary distribution $\beta$.
In order to apply Theorem~\ref{thm:main_ext}, we need to show that $(W \cup W') W_1$ can be complete with respect to its stationary distribution. The following theorem solves this problem.

\begin{theorem}\label{thm:find_generators}
Let $\mathrsfs{A}$ be a strongly connected automaton,
and the sets $W_1,W_2$ be chosen so that the matrix of the underlying digraph of $\mathrsfs{B} = \mathrsfs{A}_c(W_1,W_2)$ is primitive.
Let $\alpha$ be a stationary distribution of the Markov chain associated with $\mathrsfs{B}$ for arbitrary positive distributions $P_1,P_2$ on
$W_1,W_2$, respectively.
Then, for each set of words $W$ which is complete for $\mathbb{R}^n$ with respect to $\alpha$, the automaton
$\mathrsfs{C} = \mathrsfs{A}_c(W_1, W \cup W_2)$ is synchronizing.
\end{theorem}
\begin{proof}
For each $0 \leq \delta < 1$ we define
$$S(\delta) = ((1-\delta)[P_2] + \frac{\delta}{|W|}\sum_{w \in W}[w])[P_1].$$
Clearly, $S(\delta)$ is a positive probability distribution on $(W \cup W_2)W_1$ for each $0 < \delta < 1$,
and on $W_2W_1$ for $\delta = 0$.
Because the matrix of the underlying digraph of $\mathrsfs{B}$ is primitive,
for each $0 \leq \delta < 1$ there is a unique stationary distribution $\beta(\delta)$ such that
$\beta(\delta) S(\delta) = \beta(\delta)$ or, equivalently, $\beta(\delta)$ is the
unique stochastic solution $x$ of the equation
$$x (S(\delta) - I_n) = (0,0,\dots,0).$$
Therefore $\beta(\delta) = \tilde{S}^{-1}(\delta)(1,0, \dots ,0)$,
where $\tilde{S}(\delta)$ is the invertible matrix obtained from the
matrix $S(\delta) - I_n$ by replacing the first row by the vector of
all $1$-s. Note that $\beta(0) = \alpha$, and $\beta(\delta)$ is
(component wise) continuous in $[0,1)$.

Since $W$ is complete with respect to $\alpha$,
there are words $w_1,w_2, \dots, w_n \in W$ such that the square matrix $D = (\alpha w_i )_{i \in \{1,2,\dots,n\}}$ has rank $n$.
For $0 \le \delta < 1$ define the matrix
$$D_{\delta} = (\beta(\delta) w_i)_{i \in \{1,2,\dots,n\}}$$
and consider the function $\phi(\delta) = \det(D_{\delta})$.
Since $\beta(\delta)$ is continuous in $[0,1)$, $\phi(\delta)$ is also continuous in $[0,1)$.
Since $\phi(0) = \det(D) \neq 0$, we get that $\phi(\delta') \neq 0$ for some $0 < \delta' < 1$.
Hence $W$ is complete for $\mathbb{R}^n$ with respect to $\beta(\delta')$.
Since $\beta(\delta')$ is the stationary distribution of the Markov chain defined on $(W \cup W_2)W_1$ by the positive probability distribution $S(\delta')$, by Theorem~\ref{thm:main_ext} we obtain that the automaton $\mathrsfs{C}$ is synchronizing.
\qed
\end{proof}

\subsection{Synchronizing quasi-Eulerian automata}
\label{subsec:quasi-eulerian}

Let $\alpha$ be the probability distribution on $\Sigma^{\leq d}$
induced by a probability distribution $P\colon \Sigma \mapsto \mathbb{R}^{+}$ on
the alphabet, that is, $[P_2] = \frac{1}{n}\sum_{i=0}^{d}[P]^{i}.$
Suppose that $d < \mathrm{poly}(n)$ is such that $\Sigma^{\leq d}$ is
complete for $\mathbb{R}^n$ with respect to $\alpha$.
Using the reduction procedure, we can construct a set $U$ of at most $n$ words such that
$$\langle \alpha U \rangle = \langle \alpha \Sigma^{\leq d} \rangle = \mathbb{R}^n.$$
However, $\alpha$ is not necessarily the stationary distribution for some positive probability distribution on $U$. The following lemma solves this problem.

\begin{lemma}\label{lem:about_ex_distr}
Let $W = \{ a u \mid u \in \Suff(U), a \in \Sigma \},$
where $\Suff(U)$ is the set of proper suffixes of $U$.
Then there exists a positive probability distribution on $W$ such that $\alpha$ is the corresponding stationary distribution.
\end{lemma}
\begin{proof}


Since $W$ is complete with respect to $\alpha$, following the proof of Theorem~\ref{thm:main_ext} for
each $x \in \mathbb{R}^n \setminus \langle[Q]\rangle$, there exists $w \in W$ such that $(x,\alpha[w]) \neq (x,\alpha)$.
Suppose that $w$ is a shortest word from $W$ with this property.
If $(x,\alpha[w]) > (x,\alpha)$ then we have found an extension word from $W$.
Suppose that $(x,\alpha[w]) < (x,\alpha)$.
Clearly $1 \leq |w| \leq d$, and $w = a u$ for $a \in \Sigma$ and $u \in \Sigma^{\leq d-1}$.
Since
$$(x,\alpha[u]) = (x,\alpha[P][u]) = P(a)(x,\alpha[w]) + \sum_{b \in \Sigma, b \neq a}{P(b)(x,\alpha[b u])},$$
we get that either $(x,\alpha[u]) < (x,\alpha)$ or $(x,\alpha[b u]) > (x, \alpha)$ for some $b \neq a$.
Since $w \in W$, we have $u \in \Suff(U)$ and so $b u \in W$.
If $(x,\alpha[u]) < (x,\alpha)$ then $u \neq \varepsilon$, so $u \in W$, and $u$ is a shorter word with $(x,\alpha[u]) \neq (x,\alpha)$, which contradicts the choice of $w$.
Therefore by~\cite[Theorem 13]{Berlinkov2013QuasiEulerianOneCluster}
the automaton $\mathrsfs{B} = \mathrsfs{A}_c(\{ \varepsilon \},W)$ is
synchronizing and $\alpha$ is the stationary distribution for some
probability distribution on $W$.
\qed
\end{proof}

As an application we get a polynomial algorithm for finding a reset word for the class of \emph{quasi-Eulerian} automata, a generalization of Eulerian automata.
We call an automaton $\mathrsfs{A}$ \emph{quasi-Eulerian} with respect to an integer $c \ge 0$ if it satisfies the following two conditions:
\begin{enumerate}
\item There is a subset $E_c\subseteq Q$ containing $n-c$ states such that only one of these states, say $s$,
can have incoming edges from the set $Q \setminus E_c$.

\item There exists a positive probability distribution $P$ on $\Sigma$ such that
the columns of the matrix $[P]$ that correspond to the states from $E_c\setminus\{s\}$ sum up to 1.
\end{enumerate}

Within this definition, for $c=0$ we get so-called \emph{pseudo-Eulerian} automata,
and if additionally $P$ is uniform on $\Sigma$, then we get Eulerian automata.
The upper bound $1+(n-2)(n-1)$ on the reset thresholds of Eulerian automata was found by Kari~\cite{Kari2003Eulerian},
and extended to the class of pseudo-Eulerian automata by Steinberg~\cite{Steinberg2011AveragingTrick}.
These results were generalized in~\cite[Corollary~11]{Berlinkov2013QuasiEulerianOneCluster}
by showing the upper bound $2^c (n-c+1)(n-1)$ for the class of quasi-Eulerian automata with respect to a non-negative integer $c$.
The following theorem gives a polynomial time algorithm for finding reset words satisfying these bounds.

\begin{theorem}\label{thm:quasi_eulerian_alg}
Given a synchronizing automaton $\mathrsfs{A}$ which is quasi-Eulerian with respect to an integer $c \ge 0$,
there is a polynomial time algorithm for finding a reset word of length at most:
\begin{align*}
\begin{cases}
  2^c (n-c+1)d & \text{ if } c>0;\\
  1+(n-2)d & \text{ if } c=0,
\end{cases}
\end{align*}
where $d \leq n-1$ is the smallest integer such that $\Sigma^{\leq d}$ is complete.
\end{theorem}
\begin{proof}
First we need to calculate a stationary distribution $\alpha$, which has $n-c$ equal entries. For this purpose, for each of the $\binom{n}{n-c}$ ways of choosing the set $E_c$ containing $n-c$ states, we find a solution of the following task of linear programming:
\begin{align*}
\begin{cases}
  \alpha [P] = \alpha,\\
  ([Q], \alpha ) = 1,\\
  \alpha_p = \alpha_q & \text{ for each } p \in E_c,\\
  P(a) > 0 & \text{ for each } a \in \Sigma;
\end{cases}
\end{align*}
with the variable set
$$\{ P(a) \mid a \in \Sigma\}, \{ \alpha_p \mid p \in Q \},$$
and $q$ is an arbitrary state from $E_c$.
If there is a solution $(\alpha,P)$, then $\alpha$ is the stationary distribution for the positive probability distribution $P$ on the alphabet and it has at least $n-c$ equal entries.
Since $\binom{n}{n-c}$ is polynomial and linear programming is solvable in polynomial time, such a solution can be found in polynomial time.

Next, according to the reduction procedure for $\alpha$-completeness we can find
a polynomial set of words $W' \subseteq \Sigma^{\leq d}$ which is
complete for $\mathbb{R}^n$ with respect to $\alpha$. Due to
Lemma~\ref{lem:about_ex_distr} we can change the set $W'$ to a set
$W$ of polynomial size preserving the stationary distribution
$\alpha$ and then use the greedy extension algorithm to find a reset
word of the proposed lengths.
\qed
\end{proof}

\subsection{Synchronizing quasi-one-cluster automata}
\label{subsec:quasi-one-cluster}

The \emph{underlying digraph} of a letter $a \in \Sigma$ is the digraph with edges labeled by $a$.
Every connected component, called \emph{cluster}, in the underlying digraph of a letter has exactly one cycle, and possible some trees rooted on this cycle.
An automaton $\mathrsfs{A}=(Q,\Sigma,\delta)$ is called \emph{one-cluster} if there is a letter $a \in \Sigma$ whose underlying digraph has only one cluster.
An automaton $\mathrsfs{A}$ is \emph{quasi-one-cluster} with respect to an integer $c \ge 0$ if it has a letter whose underlying digraph has a cluster such that there are at most $c$ states in the cycles of all other clusters. Clearly, one-cluster automata are quasi-one-cluster with respect to $c=0$.
An automaton $\mathrsfs{A}$ is \emph{circular} is it has a letter whose underlying digraph consists of only one cycle of length $n$.

The \v{C}ern\'{y} conjecture was proved for \emph{circular} automata \cite{Dubuc1998}, and for one-cluster automata with prime length cycle \cite{Steinberg2011OneClusterPrime}. Also, quadratic bounds for the reset thresholds in the general case of one-cluster automata were presented \cite{BP2009QuadraticUpperBoundInOneCluster,BBP2011QuadraticUpperBoundInOneCluster,Steinberg2011AveragingTrick,CarpiDAlessandro2013IndependendSetsOfWords}.
In~\cite{Berlinkov2013QuasiEulerianOneCluster} the upper bound $2^c (2n-c-2)(n-c + 1)$ was proved for quasi-one-cluster with respect to $c$.

The following theorem gives a polynomial algorithm finding a reset word
for quasi-one-cluster automata, whose length is of the mentioned
bounds.

\begin{theorem}\label{thm:quasi_one-cluster_alg}
Let $\mathrsfs{A}$ be a synchronizing automaton that is quasi-one-cluster with respect to a letter $a$ and $c \ge 0$.
Let $C$ be the largest cycle of $a$ and $h$ be the maximal height of the trees labeled by $a$.
Let $W_1 = \{ a^{h+i} \mid i \in \{0,\ldots,|C|-1\}\}$.
Then there is a polynomial algorithm for finding a reset word for $\mathrsfs{A}$ of length at most
\begin{align*}
\begin{cases}
  2^c (2n - c)(n - c + 1) & \text{ if } c > 0;\\
  1 + (2n - r)(n - 2) & \text{ if } c = 0,
\end{cases}
\end{align*}
where $r$ is the smallest dimension of $\langle W_1 \beta \rangle$ for $\beta \in V_{C} \setminus \langle[C]\rangle$.
In particular, if $|C|$ is prime then $r = |C|$.
\end{theorem}
\begin{proof}
We can assume that $\mathrsfs{A}$ is strongly connected; otherwise,
we can use the same technique as in Theorem~\ref{thm:about_threshold}. 

Let us define $W_2 = \Sigma^{\leq n - r + 1}$ for $c = 0$ and $W_2 = \Sigma^{\leq n-1}$ otherwise.
It is proved in \cite{Steinberg2011OneClusterPrime} that for one-cluster automata 
each non-trivial subset of $S \subseteq C$ can be extended to a bigger one by a word from $W_2 W_1$.
Hence due to the greedy extension algorithm the induced automaton is synchronizing and $W_2 W_1$ is complete for $V_C$ with respect to any stochastic vector from $V_C$.
Thus in both cases we get that $W_2W_1$ is complete for $V_{Q.a^h}$. Using the reduction procedure $W_2$ can be replaced with a polynomial set of words $W$ while keeping the maximal length of words.

Let $\beta$ be the stationary distribution for some positive distribution on $W_2W_1$. Then $\beta_p > 0$ if and only if $p$ is a cycle state and $\beta_p = \beta_q$ for each $p,q \in C$.
Clearly $\DS(\beta) \leq 2^c (|C|+1)$ if $c>0$, and $\DS(\beta) = |C|-1$ if $c=0$.
According to Theorem~\ref{thm:main_ext} the automaton $\mathrsfs{B}
= \mathrsfs{A}_c(W_1,W)$ is synchronizing and we get that
\begin{align*}
\rt(\mathrsfs{A}) \leq
\begin{cases}
    h + 2^c (|C| + 1)(h + |C| + n) & \text{ if } c > 0;\\
    1 + (h + |C| + n - r)(n - 2) & \text{ if } c = 0.
\end{cases}
\end{align*}
Since the worst case appears when $|C|=n-c$ and $h=0$, the bound follows.
Since $W_1$ and $W$ have polynomial size, a reset word of this bound can be found by the greedy extension algorithm in polynomial time.
\qed
\end{proof}

\begin{remark}
The algorithm of Theorem~\ref{thm:quasi_one-cluster_alg} works also for the bounds from~\cite{CarpiDAlessandro2013IndependendSetsOfWords} for one-cluster automata. This follows in the same way as referring to~\cite{Steinberg2011OneClusterPrime} in the proof of the theorem.
\end{remark}


\section{The \v{C}ern\'{y} conjecture and random automata}
\label{sec:cerny_random}

Using the new bound, we can extend the class of automata for which
the \v{C}ern\'{y} conjecture is proven. In particular, we can
improve the result from~\cite{Pin1972Utilisation}, where the
\v{C}ern\'{y} conjecture is proven for automata with a letter of
rank at most $1+\log_2 n$.

\begin{corollary}
Let $\mathrsfs{A} = (Q,\Sigma,\delta)$ be a synchronizing automaton.
If there is a letter of rank
$$r \le \sqrt[3]{6n-6},$$
then $\mathrsfs{A}$ satisfies the \v{C}ern\'{y} conjecture.
\end{corollary}
\begin{proof}
Assume that $r \ge 3$. Using Theorem~\ref{thm:about_threshold} with $d=n-1$ and $|w|=1$ we obtain
the bound $\rt(\mathrsfs{A}) \le n(\frac{r^3 - r}{6}-1)+1$. Then
using $r \le \sqrt[3]{6n-6}$ we obtain
$$\rt(\mathrsfs{A})
< n\left(\frac{r^3}{6}-1\right)+1 \le
n\left(\frac{6n-6}{6}-1\right)+1 = (n-1)^2.
$$
If $r \le 3$ then the bound of Theorem~\ref{thm:about_threshold} is $1+n(r-1)^{2}$, which is not larger than $(n-1)^2$ for $n \ge 6$. For $n \le 5$ the \v{C}ern\'{y} conjecture has been verified \cite{KS2014SynchronizingAutomataWithLargeResetLengths}.
\qed
\end{proof}

Another corollary concerns random synchronizing automata.
We consider the uniform distribution $P_s$ on all synchronizing binary automata with $n$ states, which is formally defined by
$P_s(\mathcal{A}) = P(\mathcal{A}) / P_n$, where $P$ is the uniform distribution on all $n^{2n}$ binary automata, and $P_n$ is the probability that a uniformly random binary automaton is synchronizing.
It is known that $P_n$ tends to $1$ as $n$ goes to infinity (\cite{Berlinkov2013OnTheProbabilityToBeSynchronizable,Nicaud2014FastSynchronizationOfRandomAutomata}).

Given an arbitrary small $\varepsilon>0$ and $n$ large enough, Nicaud~\cite{Nicaud2014FastSynchronizationOfRandomAutomata} proved that a random binary automaton with $n$ states has a word of
\begin{enumerate}
\item length $n^{3/4+3\varepsilon}(1+o(1))$ and rank at most $n^{1/4+2\varepsilon}$ with probability at least $1-O(\exp(-n^\varepsilon/4))$,
\item length $n^{7/8+7\varepsilon}(1+o(1))$ and rank at most $n^{1/8+4\varepsilon}$ with probability at least $1-O(n^{-1/4+3\varepsilon})$,
\end{enumerate}

The following corollary is a consequence of these results and our Theorem~\ref{thm:about_threshold}.

\begin{corollary}\label{cor:cerny_random}
For any $\varepsilon>0$ and $n$ large enough, with probability at least $1 - O(\exp(n^{-\varepsilon/4}))$, a random $n$-state automaton with at least two letters has a reset word of length at most $n^{7/4+6\varepsilon}(1+o(1))$, and so satisfies the \v{C}ern\'{y} conjecture.
Moreover, the expected value of the reset threshold of a random synchronizing binary automaton is at most $n^{3/2+o(1)}$.
\end{corollary}
\begin{proof}
Because a random binary automaton is synchronizing with high probability, the probabilities in~(1) and~(2) remain asymptotically at least the same for a random binary synchronizing automaton.

Now, by applying our Theorem~\ref{thm:about_threshold} (with $d=n-1$) to~(1) and~(2) we get that a random binary synchronizing automaton has a reset word of
\begin{enumerate}
\item length $n^{7/4+6\varepsilon}(1+o(1))$ with probability at least $1-O(\exp(-n^\varepsilon/4))$,
\item length $n^{11/8+12\varepsilon}(1+o(1))$ with probability at least $1-O(n^{-1/4+3\varepsilon})$.
\end{enumerate}
Claim~(1) is the first statement of the corollary.

Calculating an upper bound of the average of (1), (2) and the general cubic bound applied to the rest of automata we get:
\begin{eqnarray*}
& & n^{11/8+12\varepsilon}(1+o(1)) + \\
& & n^{7/4+6\varepsilon}(1+o(1)) \cdot O(n^{-1/4+3\varepsilon}) + \\
& & (n^3-n)/6 \cdot O(\exp(-n^\varepsilon/4)) \\
& = & O(n^{11/8+12\varepsilon}) + O(n^{6/4+9\varepsilon}) \\
& = & O(n^{3/2+12\varepsilon}).
\end{eqnarray*}
\qed
\end{proof}




\section{Synchronizing finite prefix codes}
\label{sec:prefix_codes}

A \emph{finite prefix code} (Huffman code) $\mathcal{T}$ is a set of $N$ ($N > 0$)
non-empty words $\{w_1, \ldots, w_N\}$ from $\Sigma^*$, such that no
word in $\mathcal{T}$ is a prefix of another word in $\mathcal{T}$.
A finite prefix code $\mathcal{T}$ is \emph{maximal} if adding any
word $w \in \Sigma^*$ to $\mathcal{T}$ does not result in a finite
prefix code. We consider only maximal prefix codes.
A \emph{reset word} for the code $\mathcal{T}$ is a
word $w$ such that for any $u \in \Sigma^*$ the word $uw$ is a
sequence of words from $\mathcal{T}$.

One can easily see that a finite prefix code corresponds naturally to a DFA called the \emph{decoder}, whose states are proper prefixes of words from this code \cite{BiskupPlandowski2009HuffmanCodes}.
Formally, for a finite prefix code $\mathcal{T}$ we have the corresponding
\emph{decoder} $\mathrsfs{A}_\mathcal{T}$, which is the DFA
$(Q,\Sigma,\delta)$ with $Q = \{q_v \mid v\text{ is a proper prefix of a word in }\mathcal{T}\},$ and $\delta$ defined as follows:
$$\delta(q_v,a) = \begin{cases}
q_{va} & \text{if $va \not\in \mathcal{T}$;}\\
q_\varepsilon & \text{otherwise.}\\
\end{cases}$$
If for an edge from a state $q_v$ to the root $q_\varepsilon$ we
assign an output symbol associated with the word $q_v$, the decoder
can read a compressed input string and produce the decompressed
output according to the code $\mathcal{T}$. Observe that a reset
word $w$ for $\mathcal{T}$ is a reset word for the decoder
$\mathrsfs{A}_\mathcal{T}$, and $Q.w = \{q_\varepsilon\}$.
The decoder $\mathrsfs{A}_\mathcal{T}$ naturally corresponds to a
rooted $k$-ary tree. We say that $q_\varepsilon$ is the \emph{root}
state. The \emph{level} of a state $q_v \in Q$ is $|v|$, which is also the length of the shortest path from $q_\varepsilon$ to $q_v$ in the decoder DFA. The \emph{height} of $\mathrsfs{A}_\mathcal{T}$ is the maximal level of
the states in $Q$; this is also the maximal length of words from
$\mathcal{T}$.

\begin{remark}
If $N = |\mathcal{T}|$ and $k = |\Sigma|$, then the number $n$ of
states of $\mathrsfs{A}_\mathcal{T}$ is $\frac{N-1}{k-1}$.
Note that it does not depend on the length of the words in the code.
\end{remark}

In~\cite{Biskup2008HuffmanCodes,BiskupPlandowski2009HuffmanCodes} Biskup and Plandowski gave an $O(nh \log n)$ upper bound for the reset thresholds of binary decoders, where $h$ is the maximum length of a word from the code. Since $h$ can be linear in terms of $n$, this is an $O(n^2 \log n)$ general bound. Later, it was improved to $O(n^2)$ in~\cite{BBP2011QuadraticUpperBoundInOneCluster}. However, in the worst case, only decoders with a reset threshold in $\varTheta(n)$ are known \cite{BiskupPlandowski2009HuffmanCodes}, and it was conjectured that every synchronizing decoder possess a synchronizing word of length $O(n)$. Thus, there was a big gap between the upper and lower bounds for the worst case.

The following lemma is a simple generalization of~\cite[Lemma 14]{BiskupPlandowski2009HuffmanCodes} to $k$-ary decoders.
\begin{lemma}\label{lem:rank_log_word}
Let $\mathrsfs{A}_\mathcal{T}=(Q,\Sigma,\delta)$ be the $n$-state
$k$-ary synchronizing decoder of a finite prefix code $\mathcal{T}$.
There is a word $w$ of rank $r \le \lceil\log_{k}{n}\rceil$ and
length $r$.
\end{lemma}
\begin{proof}
For a word $w$, we define
$$Q(w) = \{q_v.w \mid q_v\in Q\text{ such that no prefix of $w$ maps $q_v$ to $q_\varepsilon$}\}.$$
Consider $r > 0$.
Observe that for two distinct words $w_1, w_2$ of the same length $r$ the sets $Q(w_1)$ and $Q(w_2)$ are disjoint. Also the states in $Q(w)$ are of level at least $r+1$. If for all words of length $r$ the sets $Q(w)$ are non-empty, then there are at least $k^r$ states in $Q$ of level at least $r+1$, because there are $k^r$ different words of length $r$.
Then $k^r+r+1 \le n$ and $r < \log_{k}{n}$.
Hence, if $r = \lceil \log_{k}{n}\rceil$ then there exists a word
$w$ with the empty $Q(w)$. Since any state is mapped to
$q_\varepsilon$ by a prefix of $w$, the rank of $w$ is at most
$|w|=r$. \qed
\end{proof}

Since there exists a short word of small rank $r$, we can apply
Theorem~\ref{thm:about_threshold} to improve the general upper
bounds for the reset threshold of decoders.

\begin{corollary}\label{cor:decoder_bounds}
Let $\mathrsfs{A}_\mathcal{T}=(Q,\Sigma,\delta)$ be the $n$-state
$k$-ary synchronizing decoder of a finite prefix code $\mathcal{T}$, and let $r =
\lceil\log_{k}{n}\rceil$. Then
\begin{enumerate}
\item $\rt(\mathrsfs{A}_\mathcal{T}) \leq \begin{cases}
2 + (r+n-1)(\frac{r^3-r}{6}-1) & \text{if $r \ge 4$};\\
2 + (r+n-1)(r-1)^2 & \text{if $r \le 3$}.\\
\end{cases}$
\item Any pair of states from $Q$ is compressible by a word of length at most $$r+(r+n-1)\frac{r^2-r}{2}.$$
\end{enumerate}
\end{corollary}
\begin{proof}
For Claim~1 we apply Theorem~\ref{thm:about_threshold} with $w$
being the word of rank at most $r$ and length at most $r$ from
Lemma~\ref{lem:rank_log_word}, and $d = n-1$. This gives
$(r+n-1)(\frac{r^3 - r}{6})-(n-1) = r+(r+n-1)(\frac{r^3 - r}{6}-1)$ for $r \ge 4$.

We can slightly refine the bound by Pin's result \cite[Proposition
5]{Pin1972Utilisation}, which states that if we can compress $Q.w$,
then a shortest compressing word for $Q.w$ has length at most
$|w|+n-|Q.w|+1$. Thus if $|Q.w|=r$ this is $n+1$, and we end up with
$$r-(r+n-1)+(n+1)+(r+n-1)\left(\frac{r^3 - r}{6}-1\right) = 2 +
(r+n-1)\left(\frac{r^3-r}{6}-1\right).$$
Similar calculation applies when $r \le 3$.

Claim~2 follows directly from Theorem~\ref{thm:about_threshold}.
\qed
\end{proof}

If the size $k$ of the alphabet is fixed,
Corollary~\ref{cor:decoder_bounds} yields $O(n \log^3 n)$ upper
bound for the reset threshold, and $O(n \log^2 n)$ upper bound for
the length of a word compressing a pair of states of a decoder.

Note that the word $w$ from Lemma~\ref{lem:rank_log_word} can be
easily computed in $O(n^2)$ time, since there are $O(n)$ words of
length at most $\lceil\log_{k}{n}\rceil$. Then a reset word within
the bound of Corollary~\ref{cor:decoder_bounds} can be computed in polynomial time by the algorithm discussed in Section~\ref{sec:algorithms}.

\subsection{Random binary decoders}

By a \emph{uniformly random $n$-state decoder} we understand a decoder chosen uniformly at random from the set of all $n$-state decoders.
We consider here random binary decoders.
In~\cite{FJTZ2003} it was proved that a uniformly random binary $n$-state decoder is synchronizing with a probability that tends to $1$ as $n$ goes to infinity.
Since every binary $n$-state decoder correspond to a binary tree with $n+1$ leaves, the number of all such decoders is the $n$-th Catalan number.
Note that it is known that the average height of a binary $n$-state decoder is asymptotically $\varTheta(\sqrt{n})$ \cite{Knuth2005ArtOfComputerProgrammingIV}.

\begin{theorem}
The expected reset threshold of a uniformly random synchronizing binary $n$-state decoder is at most $O(n \log n)$.
\end{theorem}
\begin{proof}
Let $\mathcal{T}$ be the code (set of codewords) of a uniformly random synchronizing binary $n$-state decoder.
Let $a$ be the first letter of the alphabet.
First we show that the length of the left-most branch of the decoder is at most logarithmic with high probability. In other words, the unique word from $a^{*} \cap \mathcal{T}$ has length $\ell \in O(\log n)$ with high probability.
Because a uniformly random binary $n$-state decoder is synchronizing with high probability, these probabilities transfer to a uniformly random synchronizing binary $n$-state decoder.
Note that if $a^\ell$ is in the code, then the the letter $a$ in the decoder is one-cluster with the cycle of length $\ell$.
Then we apply the upper bound $O(\ell n)$ on the reset thresholds of one-cluster automata whose the length of the cycle is $\ell$ \cite{BBP2011QuadraticUpperBoundInOneCluster}.

From the proof of \cite[Lemma 5.9]{FJTZ2003} we have that the fraction of decoders whose code contains $a^\ell$ (for $1 \le \ell \le n-1$) is equal to:
$$\frac{\ell (n+1)(n)\dots(n-\ell+1)}{(2n-1)(2n-2)\dots(2n-\ell-1)}.$$

Let $d$ be such that $3 \le d \le n-1$. Then 
$$\frac{n-d+1}{2n-d-1} = \frac{1}{2} - \frac{d-3}{2(2n-d-1)} \le 1/2$$ 
Hence for $\ell \geq 3$, we have 
$$\frac{\ell (n+1)(n)\dots(n-\ell+1)}{(2n-1)(2n-2)\dots(2n-\ell-1)} \le \frac{\ell}{2^{\ell+1}}.$$
It follows that the probabilities that $a^\ell$ is in the code for $\ell \ge 2\log_{2}{n}$ is at most $O(1/n)$.


For the case $\ell < 2\log_{2}{n}$, we use the upper bound $O(\ell n)$ from \cite{BBP2011QuadraticUpperBoundInOneCluster} and for $\ell \ge 2\log_{2}{n}$ with probability $O(1/n)$ we use the general upper bound $O(n \log^3 n)$ for decoders from Corollary~\ref{cor:decoder_bounds}.
Summing up these cases yield our upper bound on the expected reset threshold of the decoder:
$$O(n \log n) + O(n \log^3 n) \cdot O(1/n) = O(n \log n).$$
\qed
\end{proof}

As in the general case, a reset word of average length $O(n \log n)$ can be computed in polynomial time. This can be done using the algorithm discussed in Subsection~\ref{subsec:quasi-one-cluster}, which finds a reset word within the bounds for (quasi-)one-cluster automata.

\subsection{Lower bounds}

Biskup and
Plandowski~\cite{Biskup2008HuffmanCodes,BiskupPlandowski2009HuffmanCodes}
presented a series of binary $n$-state decoders with the reset
threshold $2n-5$ for even $n$ and $2n-7$ for odd $n$. However, only
binary decoders were studied. Here we present a series of $k$-ary
decoders for every $k \ge 3$ with large reset thresholds. This shows
that, in the worst case, also for arbitrary large
alphabets a decoder can have the reset threshold in $\varTheta(n)$.

\begin{figure}[ht]
\unitlength 7pt \scriptsize
\begin{center}\begin{picture}(32,34)(0,0)
\gasset{Nh=2,Nw=6.5,Nmr=1.5,ELdist=0.5,loopdiam=2}
\node[Nframe=n](in)(18,32){}
\node(0)(0,32){$0$}\imark(0)
\node(1)(0,28){$1$}
\node(2)(8,28){$2$} \node[Nframe=n](dotsv1)(18,28){$\dots$}
\node(k)(28,28){$k$} \drawedge(0,1){$a_1$}
\drawedge[curvedepth=1,ELpos=70](0,2){$a_2$}
\drawedge[curvedepth=2,ELpos=70](0,k){$a_k$}

\node[Nframe=n](dotsh1)(0,24){$\dots$} \drawedge[dash={.1
.2}{.2}](1,dotsh1){$a_1$}

\node(k+1i)(0,20){$(k+1)i$} \node(k+1i+1)(0,16){$(k+1)i+1$}
\node(k+1i+2)(8,16){$(k+1)i+2$}
\node[Nframe=n](dotsv2)(18,16){$\dots$}
\node(k+1i+k)(28,16){$(k+1)i+k$} \drawedge[dash={.1
.2}{.2}](dotsh1,k+1i){$a_1$} \drawedge(k+1i,k+1i+1){$a_1$}
\drawedge[curvedepth=1,ELpos=70](k+1i,k+1i+2){$a_2$}
\drawedge[curvedepth=2,ELpos=70](k+1i,k+1i+k){$a_k$}

\node[Nframe=n](dotsh2)(0,12){$\dots$} \drawedge[dash={.1
.2}{.2}](k+1i+1,dotsh2){$a_1$}

\node(k+1l)(0,8){$(k+1)\ell$} \drawedge[dash={.1
.2}{.2}](dotsh2,k+1l){$a_1$}

\node[Nframe=n](dotsh3)(0,4){$\dots$} \node(n-1)(0,0){$n-1$}
\drawedge[dash={.1 .2}{.2}](k+1l,dotsh3){$a_1$} \drawedge[dash={.1
.2}{.2}](dotsh3,n-1){$a_1$}
\end{picture}\end{center}
\caption{The decoder $\mathrsfs{X}_{n,k}$ with reset threshold
$\lceil n/(k+1)\rceil$.}\label{fig:decoder_k_long_threshold}
\end{figure}
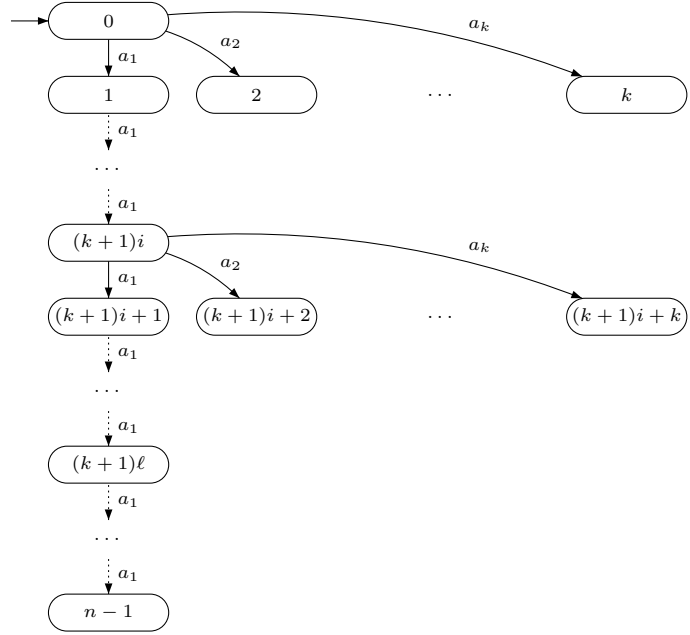

For $k \ge 3$ and $n \ge k+2$, we define $\mathrsfs{X}_{n,k} =
(Q,\Sigma,\delta)$ (shown
in~Fig.~\ref{fig:decoder_k_long_threshold}). Let $Q =
\{0,\ldots,n-1\}$ and $\Sigma = \{a_1,\ldots,a_k\}$, and let $\ell =
\lceil n/(k+1)\rceil-1$ (so $\ell \ge 1$). We define $\delta$ as
follows: For each $i$ with $0 \le i \le \ell-1$ and each $1 \le j
\le k$, if $(k+1)i+j \le n-1$ then we define: $\delta((k+1)i+1,a_j)
= (k+1)i+j$. Also for $i$ with $(k+1)\ell \le i \le n-2$ we define
$\delta(i,a_1) = i+1$. For all the remaining states $i$ and letters
$a_j$ we set $\delta(i,a_j) = 0$.

\begin{theorem}
The automaton $\mathrsfs{X}_{n,k}$ is synchronizing and its reset
threshold is $2\ell+2 = 2\lceil n/(k+1)\rceil$.
\end{theorem}
\begin{proof}
One verifies that the action of the word $a_k (a_1)^{2\ell} a_k$
synchronizes the automaton.

Let $w$ be a shortest reset word for the automaton. Consider the
first two letters $w_1,w_2$ of $w$. Observe that $Q.w_1$ and $Q.w_1
w_2$ contains $0$. So $Q.w_1 w_2$ also contains a state $p$ from
$\{2,\ldots,k\}$.

State $0$ is at the level $0$, and state $p$ is at the level $1$.
For all states $q < (k+1)\ell$ the action of all the letters
alternates the parity of the level of the states. Thus two such
states with an odd and an even level cannot be compressed by the
action of a single letter. So, to compress $\{0,p\}$, one of the
states must be first mapped to a state $q \ge (k+1)\ell$. The
shortest such a path is from $1$ to $(k+1)\ell$ labeled by
$(a_1)^{2^\ell-1}$. Then we need one more letter ($a_k$) to
synchronize the pair. It follows that $0$ and $p$ requires a word of
length at least $2\ell$ to be compressed. Hence, the length of $w$
is at least $2+2\ell$. \qed
\end{proof}

Using a more sophisticated construction, it is possible to modify our series and obtain decoders with
slightly larger reset thresholds, though still of order
$2n/(k+1)+O(1)$. We suppose that this order of growth is tight for $k \ge 3$ up to the constant within $O(1)$.

\section{Conclusions and open problems}

We have shown constructible upper bounds for the reset threshold, which turned out to be useful in several important cases of automata.
They are obtained using a uniform approach basing on Markov chains.
In all the cases, there exists a polynomial algorithm finding a reset word of length within the bounds.
Also, note that if the \v{C}ern\'{y} conjecture is true, then our bounds become reduced; in particular, we would get $O(n \log^2 n)$ for the reset thresholds of decoders of finite prefix codes.

The questions about tight bounds in the case of finite prefix codes and random automata remain open.
For finite prefix codes the bound $O(n)$ was conjectured. Note that for some applications it can be also important to get bounds in terms of the maximal length of the words in the code (e.g.~\cite{CarpiDAllesandro2014CernyLikeProblemsForFiniteSetsOfWords}).
There is also an interesting question about the expected reset threshold of the decoder of a random finite prefix code. So far for this case we have only the bound $O(n \log n)$, which comes from the bounds for one-cluster automata.

Also, there is the open problem of designing a polynomial algorithm finding reset words within the bound of~\cite{Dubuc1998} for circular automata.

\section*{Acknowledgments}

This work was supported by the Presidential Program ``Leading Scientific Schools of the Russian Federation'', project no.\
5161.2014.1, the Russian Foundation for Basic Research, project no.\
13-01-00852, the Ministry of Education and Science of the Russian
Federation, project no.\ 1.1999.2014/K, and the Competitiveness
Program of Ural Federal University (Mikhail Berlinkov),
and by the National Science Centre, Poland under
project number 2014/15/B/ST6/00615 (Marek Szyku{\l}a).

\bibliographystyle{plain}

\end{document}